\documentclass[11pt,a4paper]{article}
\usepackage{a4wide}

\usepackage{authblk}
\usepackage{amsthm}
\usepackage[T1]{fontenc}
\usepackage[utf8]{inputenc}
\usepackage{amsmath}
\usepackage{url}
\usepackage{hyperref}
\usepackage{graphicx}
\theoremstyle{definition}

\newtheorem{theorem}{Theorem}
\newtheorem{lemma}{Lemma}
\newtheorem{corollary}{Corollary}
\newtheorem{observation}{Observation}

\newcommand{\infF}{\mathcal{F}}

\newcommand{\SA}{\mathit{SA}}
\newcommand{\suf}[1]{\mathit{suf}_{#1}}
\newcommand{\sufp}[1]{\mathit{suf}_{#1}^{+}}
\newcommand{\LF}{\mathit{LF}}
\newcommand{\ED}{\mathsf{ed}}
\newcommand{\MSS}{\mathsf{MS}_{\mathtt{sub}}}
\newcommand{\MSSF}{\mathit{MS}_{\mathtt{sub}}}
\newcommand{\MSD}{\mathsf{MS}_{\mathtt{del}}}
\newcommand{\MSI}{\mathsf{MS}_{\mathtt{ins}}}
\newcommand{\AOS}{\mathsf{AOS}}
\newcommand{\AOSF}{\mathit{AOS}}
\newcommand{\maxsuf}{\mathsf{maxsuf}}

\title{Edit and Alphabet-Ordering Sensitivity of Lex-parse}
\date{}
\author[1]{Yuto~Nakashima}
\author[2,4]{Dominik~K{\"{o}}ppl}
\author[3]{Mitsuru~Funakoshi}
\author[1]{Shunsuke~Inenaga}
\author[4]{Hideo~Bannai}
\affil[1]{Department of Informatics, Kyushu University}
    \affil[ ]{\texttt{\{nakashima.yuto.003, inenaga.shunsuke.380\}@m.kyushu-u.ac.jp}}
\affil[2]{Department of Computer Science and Engineering, University of Yamanashi}
  \affil[ ]{\texttt{dkppl@yamanashi.ac.jp}}
\affil[3]{NTT Communication Science Laboratories}
  \affil[ ]{\texttt{mitsuru.funakoshi@ntt.com}}
\affil[4]{M\&D Data Science Center, Tokyo Medical and Dental University}
  \affil[ ]{\texttt{hdbn.dsc@tmd.ac.jp}}
\begin{document}
\maketitle
\begin{abstract}
    We investigate the compression sensitivity [Akagi et al., 2023] of lex-parse [Navarro et al., 2021]
    for two operations: (1) single character edit and (2) modification of the alphabet ordering,
    and give tight upper and lower bounds for both operations.
    For both lower bounds, we use the family of Fibonacci words.
    For the bounds on edit operations,
    our analysis makes heavy use of properties of the Lyndon factorization of Fibonacci words to characterize the structure of lex-parse.
\end{abstract}
\section{Introduction}
Dictionary compression is a scheme of lossless data compression that is very effective, especially for highly repetitive text collections.
Recently, various studies on dictionary compressors and repetitiveness measures have attracted much attention in the field of stringology
(see~\cite{DBLP:journals/corr/abs-2004-02781,repetitiveness_Navarro21a} for a detailed survey).

The \emph{sensitivity}~\cite{AKAGI2023104999} of a string compressor/repetitiveness measure $c$ is defined as the maximum gap in the sizes of $c$ for a text $T$ and for a single-character edited string $T'$, which can
represent the robustness of the compressor/measure w.r.t. small changes/errors of the input string.
Akagi et al.~\cite{AKAGI2023104999} gave upper and lower bounds on the worst-case multiplicative sensitivity of various compressors and measures including the Lempel--Ziv parse family~\cite{LZ77,LZ78}, the run-length encoded Burrows-Wheeler transform (RLBWT)~\cite{Burrows94ablock-sorting}, and the smallest string attractors~\cite{KempaP18}.

On the other hand, some structures built on strings, including the output of text compressors, can also depend on the order of the alphabet, meaning that
a different alphabet ordering can result in a different structure for the same string.
Optimization problems of these kinds of structures have recently been studied, e.g., for the RLBWT~\cite{BentleyGT20}, the RLBWT based on general orderings~\cite{giancarlo23new}, or the Lyndon factorization~\cite{GibneyT21}.
Due to their hardness, efficient exact algorithms/heuristics for minimization have been considered~\cite{DBLP:conf/dcc/CenzatoGLR23,DBLP:journals/ipl/ClareD19,DBLP:conf/gecco/ClareDMZ19,DBLP:conf/ppsn/MajorCDMGZ20,DBLP:journals/corr/abs-2401-16435}.

This paper is devoted to the analysis of the sensitivity of lex-parse.
Lex-parse~\cite{NavarroOP21} is a greedy left-to-right partitioning of an input text $T$ into phrases,
where each phrase starting at position $i$ is $T[i..i+\max\{1,\ell\})$ and $\ell$ is the longest common prefix between $T[i..n]$ and its lexicographic predecessor $T[i'..n]$ in the set of suffixes of $T$.
Each phrase can be encoded by a pair $(0, T[i])$ if $\ell=0$, or $(\ell, i')$ otherwise.
By using the lex-parse of size $v$ of a string,
we can represent the string with $v$ derivation rules.%
\footnote{We stick to the convention to denote the size of lex-parse by $v$ as done in literature such as~\cite{NavarroOP21,repetitiveness_Navarro21a}.}
Lex-parse was proposed as a new variant in a family of ordered parsings that is considered as a generalization of the Lempel--Ziv parsing and a subset of bidirectional macro schemes~\cite{StorerS78}.
We stress that lex-parse can have much fewer factors than the Lempel--Ziv parsing; for instance the number of Lempel--Ziv factors of the $k$-th Fibonacci word is $k$ while 
we have only four factors for lex-parse regardless of $k$ (assuming that $k$ is large enough)~\cite{NavarroOP21}.
Besides having potential for lossless data compression, it helped to gain new insights into string repetitiveness:
For instance, a direct relation $v \in O(r)$ between $v$ and the size $r$ of the RLBWT, one of the most important dictionary compressors, holds~\cite{NavarroOP21}.
Hence, combinatorial studies on lex-parse can lead us to further understanding in string repetitiveness measures and compressors.

The contribution of this paper is twofold.
We first consider the sensitivity of lex-parse w.r.t.
edit operations, and give tight upper and lower bounds
which are logarithmic in the length of the input text.
Interestingly, lex-parse is the third measure with super-constant bounds out of (about) 20 measures~\cite{AKAGI2023104999}.
Second, we consider a new variant of sensitivity, the alphabet ordering sensitivity (AO-sensitivity) of lex-parse, defined as the maximum gap in the number of phrases of lex-parse between any two alphabet orderings,
and show tight upper and lower bounds.
For both lower bounds, we use the Fibonacci word.
Moreover, we also use properties of the Lyndon factorization~\cite{ChenFL58:_lyndon_factorization_} for the edit sensitivity to characterize the structure of lex-parse.
These insights may be of independent interest.
Properties of the Fibonacci word can contribute to the analysis of algorithm complexity. In fact, there are several results regarding lower bounds based on the Fibonacci word (i.e., \cite{DBLP:journals/algorithmica/CrochemoreR95,DBLP:conf/dlt/GiulianiILRSU23,DBLP:journals/mst/InoueMNIBT22,NavarroOP21}).

\vspace*{0.25cm}

\noindent{\bf Related work.}
Inspired by the results of Lagarde and Perifel~\cite{DBLP:conf/soda/LagardeP18},
Akagi et al.~\cite{AKAGI2023104999} pioneered the systematic study of
compression sensitivity of various measures w.r.t. edit operations.
Giuliani et al.~\cite{DBLP:conf/dlt/GiulianiILRSU23} showed an improved lower bound for the additive sensitivity of the run-length BWT.
They also use the family of Fibonacci words to obtain their lower bound.
Fujimaru et al.~\cite{DBLP:conf/cwords/FujimaruNI23} presented tight upper and lower bounds for the additive and multiplicative sensitivity of
the size of the compact directed acyclic word graph (CDAWG)~\cite{blumer87complete,crochemore97direct},
when edit operations are restricted to the beginning of the text.

\section{Preliminaries}

\subsection*{Strings}
Let $\Sigma$ be an {\em alphabet}.
An element of $\Sigma^*$ is called a {\em string}.
The length of a string $w$ is denoted by $|w|$.
The empty string $\varepsilon$ is the string of length 0.
Let $\Sigma^+$ be the set of non-empty strings,
i.e., $\Sigma^+ = \Sigma^* \setminus \{\varepsilon \}$.
For any strings $x$ and $y$,
let $x \cdot y$ (or sometimes $xy$) denote the concatenation of the two strings.
For a string $w = xyz$, $x$, $y$ and $z$ are called
a \emph{prefix}, \emph{substring}, and \emph{suffix} of $w$, respectively.
They are called a \emph{proper prefix}, a \emph{proper substring}, and a \emph{proper suffix} of $w$
if $x \neq w$, $y \neq w$, and $z \neq w$, respectively.
A proper substring that is both a prefix and a suffix of $w$ is also called a \emph{border} of $w$.
The $i$-th symbol of a string $w$ is denoted by $w[i]$, where $1 \leq i \leq |w|$.
For a string $w$ and two integers $1 \leq i \leq j \leq |w|$,
let $w[i..j]$ denote the substring of $w$ that begins at position $i$ and ends at
position $j$. For convenience, let $w[i..j] = \varepsilon$ when $i > j$.
Also,
let
$w[..i]=w[1..i]$,
$w[i..]=w[i..|w|]$
$w[i..j] = w(i-1..j] = w[i..j+1)$.
For any string $w$, let $w^1 = w$ and let $w^k = ww^{k-1}$ for any integer $k \ge 2$.
A string $w$ is said to be {\em primitive} if $w$ cannot be written as $x^k$ for any $x \in \Sigma^{+}$ and integer $k \geq 2$.
The following property is well known.

\begin{lemma}[\cite{lothaire2005applied}] \label{lem:primitive-square}
    $w$ is primitive iff $w$ occurs exactly twice in $w^2$.
\end{lemma}

A sequence of $k$ strings $w_{1}, \ldots, w_k$ is called a \emph{parsing} or a \emph{factorization} of a string $w$ if $w = w_{1} \cdots w_k$.
For any binary string $w$, $\overline{w}$ denotes the bitwise reversed string of $w$ (e.g., $\overline{aab} = baa$ over $\{a,b\}$).
Let $\prec$ denote a (strict) total order on an alphabet $\Sigma$.
A total order $\prec$ on the alphabet induces a total order on the set of strings
called the \emph{lexicographic order} w.r.t.~$\prec$, also denoted as $\prec$, i.e.,
for any two strings $x, y \in \Sigma^*$,
$x \prec y \iff x$ is a proper prefix of $y$,
or, there exists $1 \leq i \leq \min\{|x|,|y|\}$
s.t. $x[1..i) = y[1..i)$ and $x[i] \prec y[i]$.

\subsection*{Lex-parse}
The lex-parse of a string $w$ is a greedy left-to-right parsing $w=w_1,\ldots,w_v$, such that each phrase
$w_j$
starting at position $i=1+\sum_{k<j}|w_k|$ is $w[i..i+\max\{1,\ell\})$, where $\ell$ is the length of the longest common prefix between $w[i..]$ and its lexicographic predecessor $w[i'..]$ in the set of suffixes of $w$.
Each phrase can be encoded by a pair $(0, w[i])$ if $\ell=0$ (called an {\em explicit phrase}), or $(\ell, i')$ otherwise.
We will call $w[i'..]$ the {\em previous suffix} of
$w[i..]$.
The size (the number of phrases) of the lex-parse of a string $w$ will be denoted by $v(w)$
Note that a phrase starting at position $i$ is explicit if and only if $w[i..]$ is the lexicographically smallest suffix starting with $w[i]$ and thus there are $|\Sigma|$ of them.
Let $w = \mathtt{ababbaaba}$.
The lex-parse of $w$ is $\mathtt{aba}$, $\mathtt{b}$, $\mathtt{ba}$, $\mathtt{a}$, $\mathtt{b}$, $\mathtt{a}$.
Since the previous suffix of $w[1..]$ is $w[7..]$ and the longest common suffix between them is $\mathtt{aba}$, the first phrase is $\mathtt{aba}$.
In this example, the last two phrases are explicit phrases.

\subsection*{Lyndon factorizations}
A string $w$ is a {\em Lyndon word}~\cite{lyndon54:_burnside} w.r.t. a lexicographic order $\prec$,
if and only if $w \prec w[i..]$ for all $1 < i \leq |w|$, i.e.,
$w$ is lexicographically smaller than all its proper suffixes with respect to $\prec$.
The \emph{Lyndon factorization}~\cite{ChenFL58:_lyndon_factorization_} of a string $w$, denoted $\LF(w)$,
is a unique factorization $\lambda_1^{p_1}, \ldots, \lambda_m^{p_m}$ of $w$,
such that each $\lambda_i \in \Sigma^+$ is a Lyndon word,
$p_{i} \geq 1$, and $\lambda_{i} \succ \lambda_{i+1}$ for all $1 \leq i < m$.
A suffix $x$ of $w$ is said to be \emph{significant} if there exists an integer $i$
such that $x = \lambda_i^{p_i} \cdots \lambda_m^{p_m}$ and
$\lambda_{j+1}^{p_{j+1}} \cdots \lambda_m^{p_m}$ is a prefix of $\lambda_j^{p_j}$ for every $j$ satisfying $i \leq j < m$~(cf.~\cite{DBLP:journals/tcs/INIBT16}).

\subsection*{Fibonacci words}
The $k$-th (finite) Fibonacci word $F_k$ over a binary alphabet $\{a, b\}$ is defined as follows (cf.~\cite{Lothaire83}):
$F_1 = b$, $F_2 = a$, $F_k = F_{k-1} \cdot F_{k-2}$ for any $k \geq 3$.
Let $f_k$ be the length of the $k$-th Fibonacci word (i.e., $f_k = |F_k|$).
We also use the infinite Fibonacci word $\infF = \lim_{k \to \infty} F_{k}$ over an alphabet $\{a, b\}$.
We also use $G_k = F_{k-2}F_{k-1}$ which will be useful for representing relations between even and odd Fibonacci words.

\begin{lemma}[Useful properties on Fibonacci word (cf.~\cite{NavarroOP21})]\label{lem:fib-properties}
    The following properties hold for every Fibonacci word $F_k$.
    \begin{enumerate}
        \item The length of the longest border of $F_k$ is $f_{k-2}$. \label{itFibPropBorder}
        \item $F_k$ has exactly three occurrences of $F_{k-2}$ at position $1$, $f_{k-2}+1$, and $f_{k-1}+1$~(suffix) for every $k \geq 6$.\label{itFibThreeOccs}
        \item $F_k = G_k[1..f_k - 2] \cdot \overline{G_k[f_k - 1..f_k]}$.
        \item $G_k = F_k[1..f_k - 2] \cdot \overline{F_k[f_k - 1..f_k]}$.
        \item $aaa$ and $bb$ do not occur as substrings in $F_k$ for every $k$~\cite[Lemma 36]{NavarroOP21}.\label{itFibAAA}
        \item $F_k$ is primitive for every $k$ (we can easily obtain the fact by Property~\ref{itFibPropBorder}).\label{itFibPrimitivity}
    \end{enumerate}
\end{lemma}

The next lemma is also useful for our proof which can be obtained by the above properties.
\begin{lemma} \label{lem:FibFourOcc}
    For any $k\geq 8$,
    $F_{k-4}$ occurs exactly eight times in $F_k$.
\end{lemma}
\begin{proof}
    By property~\ref{itFibThreeOccs} of Lemma~\ref{lem:fib-properties},
    $F_k$ has at least eight occurrences of $F_{k-4}$
    (since the suffix occurrence of $F_{k-4}$ in the second $F_{k-2}$ and the prefix occurrence of $F_{k-4}$ in the third $F_{k-2}$ are the same position).
    Suppose to the contrary that there exists an occurrence of $F_{k-4}$ in $F_k$ that is different from the eight occurrences.
    By property~\ref{itFibThreeOccs} of Lemma~\ref{lem:fib-properties},
    the occurrence crosses the boundary of the first and the second $F_{k-2}$.
    Since $F_{k-4}$ is both a prefix and a suffix of $F_{k-2}$,
    the occurrence implies that $F_{k-4}^2$ has $F_{k-4}$ as a substring that is neither a prefix nor a suffix.
    This fact contradicts Lemma~\ref{lem:primitive-square}.
\end{proof}

\subsection*{Sensitivity of lex-parse}
In this paper, we consider two compression sensitivity variants of lex-parse.
The first variant is the sensitivity by (single character) edit operations (cf.~\cite{AKAGI2023104999}).
For any two strings $w_1$ and $w_2$, let $\ED(w_1, w_2)$ denote the edit distance between $w_1$ and $w_2$.
The worst-case multiplicative sensitivity of lex-parse w.r.t.\ a substitution is defined as follows:
\[
    \MSS(v, n) = \max_{w_1 \in \Sigma^n} \MSSF(v, w_1),
\]
where $\MSSF(v, w_1) = \max \{v(w_2)/v(w_1) \mid w_2 \in \Sigma^n, \ED(w_1, w_2) = 1 \}$.
$\MSI(v,n)$ (resp. $\MSD(v,n)$) is defined by replacing the condition $w_2 \in \Sigma^n$
with $w_2 \in \Sigma^{n+1}$ (resp. $w_2 \in \Sigma^{n-1}$).
The second variant is the sensitivity by alphabet orderings.
For any string $w$ and a lexicographic order $\prec$,
let $v(w,\prec)$ be the size of the lex-parse of $w$ under $\prec$.
We define the alphabet-ordering sensitivity of lex-parse as follows:
\[
    \AOS(v, n) = \max_{w \in \Sigma^n} \AOSF(v, w),
\]
where $\AOSF(v, w) = \max_{\prec_1,\prec_2 \in A} \{v(w,\prec_2)/v(w,\prec_1) \}$.

\section{Upper bounds}
We first give upper bounds for both operations.
We can obtain the following result by using known connections regarding the bidirectional macro scheme and lex-parse.

\begin{theorem} \label{thm:sensitivity_UB}
    $\MSS(v, n), \MSI(v, n), \MSD(v, n), \AOS(v, n) \in O(\log n)$.
\end{theorem}

\begin{proof}
    For any string $w$, let $b(w)$ be the size of the smallest bidirectional macro scheme~\cite{StorerS78}.
    Then, $v(w) \geq b(w)$ holds~\cite[Lemma~25]{NavarroOP21}.
    For any two strings $w_1$ and $w_2$ with $\ED(w_1, w_2) = 1$,
    $v(w_2) \in O(b(w_2) \log (n/b(w_2)))$~\cite[Theorem~26]{NavarroOP21} and $b(w_2) \leq 2b(w_1)$~\cite[Theorem~11]{AKAGI2023104999} hold.
    Hence, for $|w_2| \in \Theta(n)$,
    \[
        \frac{v(w_2)}{v(w_1)} \leq \frac{v(w_2)}{b(w_1)}
        \in O \left( \frac{b(w_2) \log (n/b(w_2))}{b(w_1)} \right)
        \subseteq O \left( \frac{b(w_1) \log (n/b(w_1))}{b(w_1)} \right)
        = O (\log (n/b(w_1))).
    \]
    For any alphabet order $\prec$ on $\Sigma$,
    $v(w,\prec) \in O(b(w) \log (n/b(w)))$ and $v(w,\prec) \geq b(w)$ holds.
    Hence, for any two alphabet orders $\prec_1$ and $\prec_2$,
    \[
        \frac{v(w,\prec_2)}{v(w,\prec_1)} \leq \frac{v(w,\prec_2)}{b(w)}
        \in O \left( \frac{b(w) \log (n/b(w))}{b(w)} \right) = O (\log (n/b(w))).
    \]
    These facts imply this theorem.
\end{proof} 
\section{Lower bounds for edit operations}
In this section, we give tight lower bounds for edit operations with the family of Fibonacci words.
\begin{theorem}
  $\MSS(v, n), \MSI(v, n), \MSD(v, n) \in \Omega(\log n)$.
\end{theorem}
We devote this section to show $\MSS(v, n) \in \Omega(\log n)$
since a similar argument can obtain the others.
We obtain the claimed lower bound by combining
the result of \cite{NavarroOP21} (also in Lemma~\ref{lemFourFactorsEvenN} in Section~\ref{secAOsensitivity})
stating that $v(F_{2k})$ is constant,
and the following Theorem~\ref{thm:comp-sensitivity_LB}.

\begin{theorem} \label{thm:comp-sensitivity_LB}
    For every integer $k \geq 6$,
    there exists a string $w$ of length $f_{2k}$ such that $\ED(F_{2k}, w) = 1$ and $v(w) = 2k-2$.
\end{theorem}

For any string $w$, let $w' = w[1..|w|-1], w'' = w[1..|w|-2]$.
Also, let $T_{2k}$ denote the string obtained from $F_{2k}$ by substituting the rightmost $b$ of $F_{2k}$ with $a$, i.e., $T_{2k} = F''_{2k} \cdot aa$.
We show that the lex-parse of $T_{2k}$ has $2k-2$ phrases.
More specifically, we show that the lengths of the phrases are
\[
    f_{2k-1}-1, (f_{2k-4}-1, f_{2k-5}+1, ..., f_4 - 1, f_3 + 1), 1, 2, 1.
\]
There are three types of phrases as follows: (1)~first phrase, (2)~inductive phrases, (3)~last three short phrases.
Phrases of Type (1) or Type (3) can be obtained by simple properties on the Fibonacci word.
However, phrases of Type (2) need a more complicated discussion.
We use the Lyndon factorizations of the strings to characterize the inductive phrases.
Intuitively, we show that every suffix of $T_{2k}$ that has an odd inductive phrase as a prefix
can be written as the concatenation of the odd inductive phrase and a significant suffix.

\subsection*{(1)~First phrase and (3)~Short phrases}
We start from Type~(1).
By the third property of Lemma~\ref{lem:fib-properties} and the edit operation,
$T_{2k}[1..]$ and $T_{2k}[f_{2k-2}+1..]$ have a (longest) common prefix $x = F'_{2k-1}$ of length $f_{2k-1}-1$
and $T_{2k}[1..] \succ T_{2k}[f_{2k-2}+1..]$ holds.
$T_{2k}[f_{2k-2}+1..]$ can be written as $T_{2k}[f_{2k-2}+1..] = x \cdot a$.
We show that the previous suffix of $T_{2k}[1..]$ is $T_{2k}[f_{2k-2}+1..]$.
Suppose to the contrary that there exists a suffix $y$ of $F_{2k}$ that satisfies
$T_{2k}[1..] \succ y\succ T_{2k}[f_{2k-2}+1..]$.
Since both $T_{2k}[f_{2k-2}+1..]$ and $T_{2k}[1..]$ have $x$ as a prefix,
$y$ can be written as $y = x \cdot a \cdot z_1$ or $y = x \cdot b \cdot z_2$ for some strings $z_1, z_2$.
Since $F_{2k-1}$ ends with $aab$ and thus
$aa$ is a suffix of $x$,
the existence of $y = x \cdot a \cdot z_1$ contradicts the fact that $a^3$ only occurs at the edit position.
On the other hand,
the existence of $y = x \cdot b \cdot z_2$ contradicts the fact that $x \cdot b = F_{2k-1}$ only occurs as a prefix of $T_{2k}$,
since otherwise it would violate the second property of Lemma~\ref{lem:fib-properties}.
Thus $T_{2k}[f_{2k-2}+1..]$ is the previous suffix of $T_{2k}[1..]$,
and the length of the first phrase is $|x| = f_{2k-1}-1$.

Next, we consider Type (3) phrases.
For any string $w$, let $\SA_{w}$ denote the {\em suffix array} of $w$,
where the $i$-th entry $\SA_{w}[i]$ stores the index $j$ of the lexicographically $i$-th suffix $w[j..]$ of~$w$.
Since $T_{2k}$ ends with $baaa$ and no Fibonacci word has $aaa$ as a substring, we conclude that
$\SA_{T_{2k}}[1] = f_{2k}$, $\SA_{T_{2k}}[2] = f_{2k}-1$, and $\SA_{T_{2k}}[3] = f_{2k}-2$.
In particular, $ba^3$ is the smallest suffix of $T_{2k}$ that begins with $b$.
These facts imply that $T_{2k}[f_{2k}] = a$ and $T_{2k}[f_{2k}-3] = b$ are the explicit phrases,
and $T_{2k}[f_{2k}-2..f_{2k}-1] = a^2$ between the explicit phrases is a phrase.
Thus the last three phrases are $b, a^2, a$.

\subsection*{(2)~Inductive phrases}
In the rest of this section, we explain Type~(2) phrases.
Firstly, we study the Lyndon factorization $\LF(\infF) = \ell_1, \ell_2, \ldots$ of the (infinite) Fibonacci word.
To characterize these Lyndon factors $\ell_i$, we use the string morphism $\phi$ with
$\phi(a) = aab, \phi(b) = ab$ as defined in~\cite[Proposition~3.2]{MELANCON2000137}.

\begin{lemma}[\cite{MELANCON2000137}] \label{lem:infinite_LF}
    $\ell_1 = ab, \ell_{k+1} = \phi(\ell_k)$, and $|\ell_k| = f_{2k+1}$ holds.
\end{lemma}

In order to show our result, we consider the Lyndon factorization of $F''_{2k} (= T''_{2k})$ (Lemma~\ref{lem:LF_shrinked-Fib}).
Lemmas~\ref{lem:LF_morphism} and~\ref{lem:LF_shrinked-factor} explain the Lyndon factorization of a finite prefix of the (infinite) Fibonacci word
by using properties on the morphism $\phi$.

\begin{lemma} \label{lem:LF_morphism}
    Given a string $w \in \{a,b\}^+$,
    let $\LF(w) = \lambda_1^{p_1}, \ldots, \lambda_m^{p_m}$.
    Then $\LF(\phi(w)) = \phi(\lambda_1^{p_1}), \ldots, \phi(\lambda_m^{p_m})$.
\end{lemma}
\begin{proof}
    For any two binary strings $x$ and $y$, it is clear that $\phi(x) \succ \phi(y)$ if $x \succ y$.
    In the rest of this proof, we show that $\phi(x)$ is a Lyndon word for any Lyndon word $x$ over $\{a, b\}$.
    If $|x| = 1$, then the statement clearly holds.
    Suppose that $|x| \geq 2$.
    Let $\tilde{x}$ be a non-empty proper suffix of $\phi(x)$.
    Then $\tilde{x}$ can be represented as $\tilde{x} = \phi(y)$ for some suffix $y$ of $x$,
    or $\tilde{x} = \alpha \cdot \phi(y)$ for some suffix $y$ of $x$ and $\alpha \in \{ab, b\}$.
    Since $x \prec y$, $\phi(x) \prec \tilde{x}$ if $\tilde{x} = \phi(y)$ for some suffix $y$ of $x$.
    Also in the case of $\tilde{x} = \alpha \cdot \phi(y)$, $\phi(x) \prec \tilde{x}$ holds since $x[1] = a$ (from $x$ is a Lyndon word) and $\phi(x)[1..3] = aab$.
    Thus, $\phi(x) \prec \tilde{x}$ holds for all non-empty proper suffixes $\tilde{x}$ of $x$.
    This implies that $\phi(x)$ is a Lyndon word.
    Therefore, the statement holds.
\end{proof}

\begin{lemma} \label{lem:LF_shrinked-factor}
    For every integer $i \geq 1$, $\LF(\ell_i') = \phi^{i-1}(a), \ldots, \phi^0(a)$,
    where $\ell_i' = \ell_i[..|\ell_i|-1]$.
\end{lemma}
\begin{proof}
    We prove the statement by induction on $i$.
    For the base case $i = 1$, $\LF(\ell_1') = a = \phi^0(a)$.
    Assume that the lemma holds for all $i \leq j$ for some $j \geq 1$.
    By the definition of the morphism $\phi$,
    $\ell_{j+1}' = \phi(\ell_j)' = \phi(\ell_j') \cdot a$ holds because $\ell_j$ ends with $b$.
    Also, by using Lemma~\ref{lem:LF_morphism} and the induction hypothesis,
    \[
        \LF(\ell_{j+1}') = \LF(\phi(\ell_j') \cdot a) = \LF(\phi(\ell_j')), a = \phi(\LF(\ell_j')), a
                         = \phi^{j}(a), \ldots, \phi^1(a), \phi^0(a).
    \]    
    Therefore, the lemma holds.
\end{proof}

\begin{lemma} \label{lem:LF_shrinked-Fib}
    Let $L_i$ be the $i$-th Lyndon factor of $F_{2k}''$.
    For every $k \geq 2$,
    \begin{equation*}
        L_i =
        \begin{cases}
            \phi^{i-1}(ab)   & \text{if $1 \leq i < k-1$,}       \\
            \phi^{2k-i-3}(a) & \text{if $k-1 \leq i \leq 2k-3$.}
        \end{cases}
    \end{equation*}
\end{lemma}
\begin{proof}
    By Lemma~\ref{lem:infinite_LF},
    \[
        \sum_{i=1}^{k-1} |\ell_i| = \sum_{i=1}^{k-1} f_{2i+1} = f_3 + f_5 + \cdots + f_{2k-1}
                                  = f_2 + (f_3 + f_5 + \cdots + f_{2k-1}) - 1 = f_{2k} - 1.
    \]
    Thus $\LF(F_{2k}) = \ell_1, \ldots, \ell_{k-1}, a$ holds.
    Also, $L_i = \ell_i$ holds for all $i < k-1$ since $\ell_i$ is not a prefix of $\ell_{i-1}$.
    In other words, $\LF(F_{2k}'') = \ell_1, \ldots, \ell_{k-2}, \LF(\ell_{k-1}')$.
    By Lemma~\ref{lem:LF_shrinked-factor},
    \[
        \LF(F''_{2k}) = \ell_1, \ldots, \ell_{k-2}, \phi^{k-2}(a), \ldots, \phi^0(a).
    \]
    Then the statement also holds.
\end{proof}

Moreover, we can find the specific form of suffixes characterized by the Lyndon factorization of $F''_{2k}$ as described in the next lemma.

\begin{lemma} \label{lem:significant-suffix}
    For every integer $i \geq 1$, $\phi^{i-1}(a) \cdots \phi^0(a)$ is a prefix of $\phi^{i}(a)$.
\end{lemma}
\begin{proof}
    We prove the lemma by induction on $i$.
    For the base case $i=1$, $\phi^0(a) = a$ is a prefix of $\phi^1(a) = aab$.
    Assume that the statement holds for all $i \leq j$ for some $j \geq 1$.
    For $i = j+1$,
    \[
        \phi^{j+1}(a) = \phi^{j}(\phi(a)) = \phi^{j}(aab) = \phi^{j}(a) \phi^{j}(a) \phi^{j}(b).
    \]
    By induction hypothesis, $\phi^{j}(a) = \phi^{j-1}(a) \cdots \phi^0(a) \cdot x$ for some string $x$.
    Then
    $\phi^{j+1}(a) = \phi^{j}(a) \cdot \phi^{j-1}(a) \cdots \phi^0(a) \cdot x \cdot \phi^{j}(b)$.
    This implies that the statement also holds for $i = j+1$.
    Therefore, the lemma holds.
\end{proof}

With the Lemmas~\ref{lem:LF_shrinked-Fib} and~\ref{lem:significant-suffix}, we obtain the following lemma,
which characterizes the first $k+1$ entries of the suffix array $\SA_{T_{2k}}$ of $T_{2k}$.

Firstly, we consider the order of the significant suffixes of $F''_{2k}$.

\begin{lemma} \label{lem:significant-suf-orderA}
    \(
    \SA_{F''_{2k}}[i] = f_{2k} - 1 - \sum_{j=0}^{i-1} |\phi^{j}(a)|
    \)
    for every $i \in [1..k-1]$.
\end{lemma}
\begin{proof}
    We can see that $\phi^{k-2}(a) \cdots \phi^0(a)$ is a suffix of $F''_{2k}$ by Lemma~\ref{lem:LF_shrinked-Fib}.
    Our claim is that $\phi^{i-1}(a) \cdots \phi^0(a)$ is the $i$-th lexicographically smallest suffix of $F''_{2k}$ for every $0 \leq i \leq k-1$.
    We prove the statement by induction on $i$.
    For the base case $i = 1$, $\phi^0(a) = a$ is the lexicographically smallest suffix of $F''_{2k}$.
    Assume that the statement holds for all $i \leq i'$ for some $i' \geq 1$.
    Let $\alpha$ be a suffix of $\phi^{i'+1}(a) \cdots \phi^0(a)$ such that there is no $d$ with $\alpha = \phi^{d}(a) \cdots \phi^0(a)$.
    (Otherwise we already know that $\alpha$ is a prefix of $\phi^{i'+1}(a)$, which we already processed for a former suffix array entry.)
    Then $\alpha$ can be written as $\alpha = \beta \cdot \phi^{d}(a) \cdots \phi^0(a)$ for some proper suffix $\beta$ of $\phi^{d+1}(a)$ and $d \in [0..i']$.
    Since $\phi^{i'+1}(a)$ has $\phi^{d+1}(a)$ as a prefix and there is a mismatching character between $\phi^{d+1}(a)$ and $\beta$
    (from $\beta$ is a suffix of Lyndon word $\phi^{d+1}(a)$),
    then $\phi^{i'+1}(a) \prec \beta$ holds.
    Therefore, the lemma holds.
\end{proof}

Because appending $a$'s to the end does not affect to suffix orders,
we can easily obtain the following lemma by Lemma~\ref{lem:significant-suf-orderA}.

\begin{lemma} \label{lem:significant-suf-order}
    For every $k \geq 2$,
    $\SA_{T_{2k}}[1] = f_{2k}$,
    $\SA_{T_{2k}}[2] = f_{2k}-1$, and
    $\SA_{T_{2k}}[i] = f_{2k}-1-\sum_{j=0}^{i-3} |\phi^{j}(a)|$ for any $i$ that satisfying $3 \leq i \leq k+1$.
\end{lemma}
\begin{proof}
    $T_{2k}$ ends with the suffix $aa$ such that $T[f_{2k}..] = a$ is the smallest, and $T[f_{2k}-1..]$ is the second smallest suffix of $T_{2k}$.
    All other suffixes inherit their order from $F''_{2k}$, and thus $\SA_{T_{2k}}[i] = \SA_{F''_{2k}}[i-2]$ for $i \ge 3$.
\end{proof}

We define a substring $Y_i$ and suffixes $X_i$ and $Z_i$ of $T_{2k}$ as follows:
\begin{align*}
    X_i & = T_{2k}[f_{2k}-f_{2i+4}..]~(1 \leq i \leq k-3),                  \\
    Y_i & = T_{2k}[f_{2k}-f_{2i+4}..f_{2k}-f_{2i+3}-2]~(1 \leq i \leq k-3), \\
    Z_i & = T_{2k}[f_{2k}-f_{2i+3}-1..]~(0 \leq i \leq k-3).
\end{align*}

\begin{observation} \label{obs:significant-substr}
    The following properties hold:
    \begin{align*}
        X_i & = Y_i \cdot Z_i = b \cdot T_{2i+4},                                 \\
        Y_i & = b \cdot T''_{2i+2},                            \\
        Z_i & = b \cdot \phi^{i}(a) \cdots \phi^0(a) \cdot aa.
    \end{align*}
\end{observation}

\begin{figure}[t]
    \centerline{\includegraphics[width=1.0\linewidth]{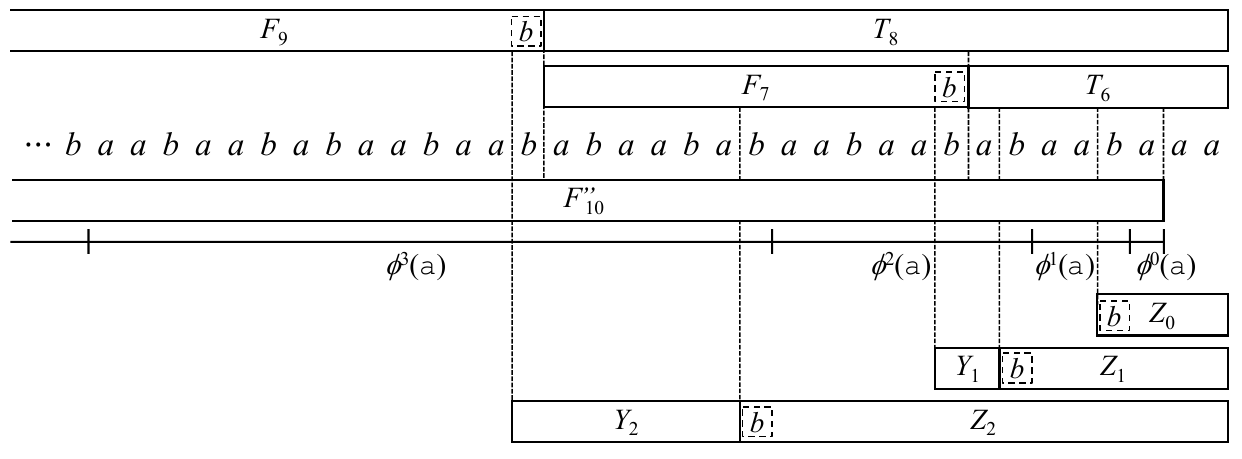}}
    \caption{
    Illustration of the characterization of the edited string $T_{2k}$ by the Lyndon factorization of $F''_{2k}$ (when $k=5$).
    }\label{fig:factorization}
\end{figure}

Notice that $|\phi^{i}(a) \cdots \phi^0(a)| = |\ell_{i+1}'| = f_{2i+3}-1$ holds.
See also Fig.~\ref{fig:factorization} for an illustration of the specific substrings.
Then Lemma~\ref{lem:significant-suf-order} implies the following corollary.

\begin{corollary} \label{coro:right-ref}
    For every integer $k \geq 2$ and $i$ satisfying $1 \leq i \leq k-2$,
    the previous suffix of $Z_i$ w.r.t. $T_{2k}$ is $Z_{i-1}$.
\end{corollary}

The largest suffix of $T_{2k}$, denoted by $\mathsf{maxsuf}$, is characterized in the following lemma.
We also use this suffix in the main lemma.
Intuitively, we show that every suffix of $T_{2k}$ that has an even inductive phrase as a prefix references a suffix that is the concatenation of a string and the maximum suffix (Lemma~\ref{lem:left-ref}).

\begin{lemma} \label{lem:maxsuf}
    The lexicographically largest suffix $\mathsf{maxsuf}$ of $T_{2k}$ is $T_{2k}[f_{2k-3}..]$.
\end{lemma}
\begin{proof}
    It is known that the lexicographically largest suffix of $F_{2k}$ is $F_{2k}[f_{2k-1}..]$
    and the lexicographically second largest suffix of $F_{2k}$ is $F_{2k}[f_{2k-3}..]$ (shown in \cite{KopplI15}).
    Namely, $\SA_{F_{2k}}[f_{2k}] = f_{2k-1}, \SA_{F_{2k}}[f_{2k}-1] = f_{2k-3}$.
    Due to the edit operation, $T_{2k}[f_{2k-3}..] \succ T_{2k}[f_{2k-1}..]$
    and the length of the longest prefix of these suffixes is $f_{2k-2}$ (see Fig.~\ref{fig:maxsuf}).
    Assume on the contrary that there is a suffix $T_{2k}[i..]$ that is lexicographically larger than $T_{2k}[f_{2k-3}..]$.
    With Property~\ref{itFibAAA} of Lemma~\ref{lem:fib-properties}, the suffix $aaa$ of $T_{2k}$ acts as a unique delimiter such that
    the suffix $T_{2k}[f_{2k-3}..]$ cannot be a prefix of any other suffixes of $T_{2k}$.
    Let $j$ be the smallest positive integer such that $T_{2k}[f_{2k-3}..f_{2k-3}+j] = a$ and $T_{2k}[i..i+j] = b$.
    If $\max(i+j, f_{2k-3}+j) < f_{2k}-1$, then $F_{2k}[f_{2k-3}..] \prec F_{2k}[i..]$ holds.
    This contradicts the fact that the lexicographically second largest suffix of $F_{2k}$ is $F_{2k}[f_{2k-3}..]$.
    We can observe that there is no $j$ with $\max(i+j, f_{2k-3}+j) \geq f_{2k}-1$ and $i>f_{2k-3}$
    since $T_{2k}[f_{2k}-1..] = aa$ does not contain $b$.
    If $\max(i+j, f_{2k-3}+j) \geq f_{2k}-1$ and $i < f_{2k-3}$, $F_{2k-3}$ has a beginning position $d$ of an occurrence satisfying $2 \leq d \leq f_{2k-3}$.
    This contradicts Lemma~\ref{lem:fib-properties}.
    Therefore, the lemma holds.
\end{proof}

\begin{figure}[t]
    \centerline{\includegraphics[width=1.0\linewidth]{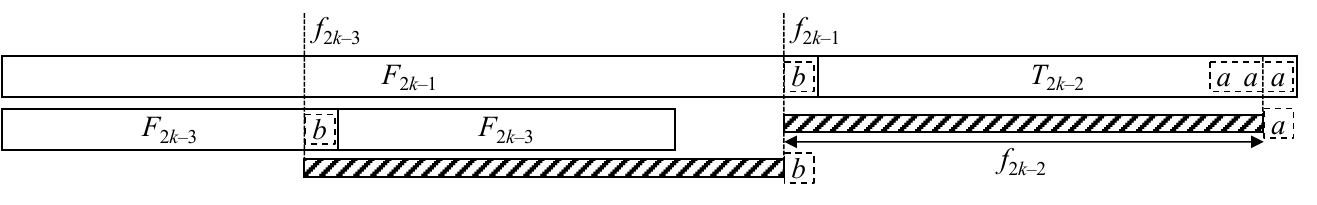}}
    \caption{%
        Illustration of $T_{2k}$ for proof of Lemma~\ref{lem:maxsuf}.
    }\label{fig:maxsuf}
\end{figure}

To prove Lemma~\ref{lem:left-ref}, we also introduce the following corollary.

\begin{corollary} [of Lemma~\ref{lem:fib-properties}] \label{coro:three-occ}
    For every $i$ satisfying $i \geq 6$,
    $T_{i}$ has exactly two occurrences of $F_{i-2}$.
    Moreover, $T_{i}$ can be written as $T_{i} = F_{i-2} \cdot F_{i-2} \cdot w$ for some string $w$.
\end{corollary}

\begin{lemma} \label{lem:left-ref}
    For every $k \geq 2$ and $i$ satisfying $1 \leq i \leq k-3$,
    the previous suffix of $X_i$ w.r.t. $T_{2k}$ is $Y_i \cdot a \cdot \mathsf{maxsuf}$.
\end{lemma}

\begin{proof}
    Firstly, we show that $Y_i \cdot a \cdot \mathsf{maxsuf}$ is a suffix of $T_{2k}$.
    It is clear from definitions and properties of Lemma~\ref{lem:fib-properties} that $T_{2k}$ can be written as
    $T_{2k} = F'_{2k-3} \cdot \mathsf{maxsuf} = F''_{2k-3} \cdot a \cdot \mathsf{maxsuf}
        = F_{2k-4} \cdot F''_{2k-5} \cdot a \cdot \mathsf{maxsuf} = F_{2k-5} \cdot F''_{2k-4} \cdot a \cdot \mathsf{maxsuf}$.
    Since $Y_{k-3} = b \cdot T''_{2k-4} = b \cdot F''_{2k-4}$ and the last character of $F_{2k-5}$ is $b$ (from $2k-5$ is odd),
    $Y_{k-3} \cdot a \cdot \mathsf{maxsuf}$ is a suffix of $T_{2k}$.
    Moreover, $F_{2i+2}$ is a suffix of $F_{2k-4}$ for every $i$ that satisfying $1 \leq i \leq k-3$.
    This implies that $Y_{i} \cdot a \cdot \mathsf{maxsuf}$ is a suffix of $T_{2k}$ for every $i$ that satisfying $1 \leq i \leq k-3$.

    Now we go back to our main proof of the lemma.
    Since $X_i = Y_i \cdot b \cdot \phi^{i}(a) \cdots \phi^0(a) \cdot aa$, $Y_i \cdot a \cdot \maxsuf \prec X_i$ holds.
    Moreover, $Y_i \cdot a \cdot \mathsf{maxsuf}$ is the lexicographically largest suffix of $T_{2k}$ that has $Y_i \cdot a$ as a prefix.
    Hence, it is sufficient to prove that $X_i$ is the lexicographically smallest suffix of $T_{2k}$ that has $Y_i \cdot b$ as a prefix.
    By Property~\ref{itFibAAA} of Lemma~\ref{lem:fib-properties}, no $bb$ occurs in $T_{2k}$, so
    every occurrence of $Y_i \cdot b$ is also an occurrence of $Y_i \cdot ba$.
    We consider occurrences of $Y_i \cdot ba$ in a suffix $X_i$.
    By Observation~\ref{obs:significant-substr} and Corollary~\ref{coro:three-occ},
    $(Y_i \cdot ba)[2..] (= F_{2i+2})$ has exactly two occurrences in $X_i (= b \cdot T_{2i+4})$.
    At the first occurrence, $(Y_i \cdot ba)[2..]$ is preceded by $b$.
    At the second occurrence, $(Y_i \cdot ba)[2..]$ is preceded by $a$.
    Hence, the rightmost occurrence of $Y_i \cdot ba$ in $T_{2k}$ is at the prefix of $X_i$.
    By combining with Lemma~\ref{lem:significant-suf-order}, 
    we can see that there is no suffix $w$ such that $Y_i \cdot ba \cdot w \prec X_i$ holds.
    Thus, $X_i$ is the lexicographically smallest suffix of $T_{2k}$ that has $Y_i \cdot b$ as a prefix.
    Therefore, the lemma holds.
\end{proof}

\subsection*{Lex-parse of $T_{2k}$}
Now we can explain the lex-parse of $T_{2k}$.
Recall that the length of the first phrase is $f_{2k-1}-1$ and the last three phrases are $b, a^2, a$.
Hence, by Lemma~\ref{lem:left-ref}, the second phrase is a prefix $Y_{k-3}$ of $X_{k-3} (= b \cdot T_{2k-2})$.
Since the remaining suffix is $Z_{k-3}$, the next phrase is a prefix $Z_{k-4}[..|Z_{k-4}|-1]$ of $Z_{k-3}$ by Corollary~\ref{coro:right-ref}.
By applying this argument repeatedly,
we can finally obtain the lex-parse of size $2k-2$ of $T_{2k}$ as follows:
\[
    F_{2k}[..f_{2k-1}-1], (Y_{k-3}, Z_{k-4}[..|Z_{k-4}|-1]), \ldots, (Y_{1}, Z_{0}[..|Z_{0}|-1]), b, a^2, a.
\]

Furthermore, we can easily obtain $v(F''_{2k} \cdot a) = 2k-2$ for a delete operation.
Consider the case for the insertion operation such that $\$$ (which is the smallest character)
is inserted to the preceding position of the last $b$.
We can then show that the lex-parse is of size $2k+1$ by a similar argument as follows:
\[
    F_{2k}[1..f_{2k-1}-2], (a \cdot Y_{k-3}, Z_{k-4}[..|Z_{k-4}|-2]), \ldots, (a \cdot Y_{1}, Z_{0}[..|Z_{0}|-2]), a, ba, \$, b, a.
\]
\section{Lower bounds for Alphabet-Ordering}\label{secAOsensitivity}
In this section,
we give tight lower bound $\AOS(v, n) \in \Omega(\log n)$ with the family of Fibonacci words.
Since $b(F_k) = 4$ for $k \ge 5$~\cite[Lemma 35]{NavarroOP21} also holds, our lower bound is tight for any $n \in \{f_i\}_{i}$.
More precisely, we prove the following theorem that determines the number of lex-parse phrases of the Fibonacci words on any alphabet ordering.
\begin{theorem}\label{thmAOLowerBound}
    For any $k \geq 6$,
    \begin{equation*}
        v(F_k,\prec) =
        \begin{cases}
            \lceil \frac{k}{2} \rceil + 1 & \text{(a) if $k$ is odd and $a \prec b$,}  \quad (Lemma~\ref{lemLogFactorsOddN})   \\
            4                             & \text{(b) if $k$ is odd and $b \prec a$,}  \quad (Lemma~\ref{lemFourFactorsOddN})  \\
            4                             & \text{(c) if $k$ is even and $a \prec b$,} \quad (Lemma~\ref{lemFourFactorsEvenN}) \\
            \lceil \frac{k}{2} \rceil + 1 & \text{(d) if $k$ is even and $b \prec a$.} \quad (Lemma~\ref{lemLogFactorsEvenN})
        \end{cases}
    \end{equation*}
\end{theorem}
Although the results for $a$ smaller than $b$ have been proven by Navarro et al.~\cite{NavarroOP21},
we here give alternative proofs for this case that leads us to the proof for the case when $b$ is smaller than $a$.

\subsection{Lex-parse with constant number of phrases}
We start with Cases~(b) and (c).
Since $F_k[f_{k}-1..f_k] = ba$ for even $k$ and $F_k[f_{k}-1..f_k] = ab$ for odd $k$,
we already know in the cases (b) and (c) that each of the last two characters forms an explicit phrase.
It is left to analyze the non-explicit phrases, where we start with the first phrase.
By Property~\ref{itFibPropBorder} of Lemma~\ref{lem:fib-properties}, $F_k$ has the border $F_{k-2}$ and thus $F_{k-2} = F_k[f_{k-1}..] \prec F_k$ could be used as the reference for the first phrase, given its length is $f_{k-2}$.
To be an eligible reference for lex-parse, we need to check that $F_k[f_{k-1}..]$ is the previous suffix of $F_k[1..]$.
However, Property~\ref{itFibThreeOccs} of Lemma~\ref{lem:fib-properties} states that there is exactly one other occurrence of $F_{k-2}$ in $F_k$, starting at $f_{k-2}+1$.
The proof of the following lemma shows that the suffix starting with that occurrence is lexicographically larger than $F_k$, and thus indeed the first phrase has length $f_{k-2}$,
and the second phrase starting with that occurrence can make use of $F_k$ as a reference for a phrase that just ends before the two explicit phrases at the end.

\begin{figure}[t]
    \centering
    \includegraphics[width=0.8\textwidth]{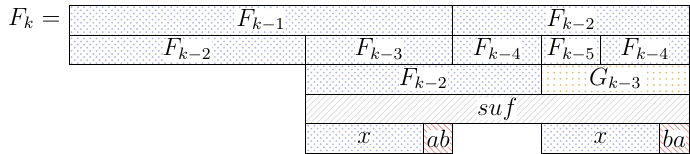}
    \caption{Illustration of the proof of Lemma~\ref{lemFourFactorsEvenN} for $k$ even. If $k$ is odd, the blocks $ab$ and $ba$ are swapped (this gives the setting Lemma~\ref{lemFourFactorsOddN}).}
    \label{figFourFactorsEvenN}
\end{figure}

\begin{lemma}\label{lemFourFactorsEvenN}
    Assume that $k \geq 6$ is even and $a \prec b$ (Case (c) of Thm.~\ref{thmAOLowerBound}).
    Then the lex-parse of $F_k$ is $F_{k-2}$,
    $F_k[f_{k-2}+1..f_k-2]$, $b$, $a$.
\end{lemma}
\begin{proof}
    $F_k$ can be represented as
    \[
        F_k = F_{k-1} \cdot F_{k-2} = F_{k-2} \cdot F_{k-3} \cdot F_{k-4} \cdot F_{k-5} \cdot F_{k-4}.
    \]
    Let $\mathit{suf} = F_{k-2} \cdot G_{k-3}$~(a suffix of $F_k$) and $x$ be the longest common prefix of $F_{k-3}$ and $G_{k-3}$.
    Then
    $F_k = F_{k-2} \cdot x \cdot ab \cdot F_{k-4} \cdot F_{k-5} \cdot F_{k-4}$ and
    $\mathit{suf} = F_{k-2} \cdot x \cdot ba$ holds
    since $k$ is even and $k-3$ is odd.
    See Fig.~\ref{figFourFactorsEvenN} for a sketch.
    This implies that the three suffixes that have $F_{k-2}$ as a prefix satisfy
    $F_{k-2} \prec F_k \prec \mathit{suf}$.
    By combining with Property~\ref{itFibThreeOccs} of Lemma~\ref{lem:fib-properties},
    the first phrase is $F_{k-2}$ and the second phrase is $F_{k-2} \cdot x$
    (which is the longest common prefix of $F_k$ and $\mathit{suf}$).
    The third phrase is $b$ which is an explicit phrase of character $b$
    since the suffix $ba$ is the lexicographically smallest suffix that has $b$ as a prefix.
    Finally, the last phrase is an explicit phrase $a$.
\end{proof}

\subsection{Lex-parse with logarithmic number of phrases}
Next, we discuss the cases that have a logarithmic number of phrases.
For any odd $k \geq 7$ and even $i$ satisfying $4 \leq i \leq k-3$, let
\begin{align*}
    \suf{i}  & = (F_i \cdot F_{i-2} \cdot F_{i-3} \cdots F_4) \cdot ab = F_{i+1}, \\
    \sufp{i} & = F_i \cdot \suf{i} = G_{i+2}.
\end{align*}
From the definitions, the following properties hold.
\begin{lemma} \label{lem:suf-suf}
    For odd $k \geq 7$, $\suf{i}$ and $\sufp{i}$ are suffixes of $F_k$, for every even $i$ with $4 \leq i \leq k-3$.
    In particular, we have
    \begin{enumerate}
        \item[(a)] $\suf{i} = F_{i+1}$ is a prefix of $\sufp{i} = G_{i+2} = F_i F_{i+1}$, and
        \item[(b)] $\sufp{i} = F_i F_{i+1} = F_{i} F_{i-1} F_{i-2} F_{i-1} = F_{i+1} G_{i} = \suf{i} \cdot \sufp{i-2}$.
    \end{enumerate}
\end{lemma}
\begin{proof}
    By definitions,
    $F_k = F_{k-2} F_{k-3} F_{k-2} = F_{k-2} F_{k-3} \suf{k-3} = F_{k-2} \sufp{k-2}$.
    Since $\suf{i-2}$ is a suffix of $\suf{i}$, the claim holds for $\suf{i}$ for every $i$.
    Writing $\suf{i} = F_{i-1} \cdot F_{i-2} \cdot (F_{i-2} \cdot F_{i-3} \cdots F_4) \cdot ab$,
    we can see that $\sufp{i-2}$ is a suffix of $\suf{i}$.
\end{proof}

We use these suffixes to characterize the lex-parse.
The following lemma shows that $\suf{i}$ is the previous suffix of $\sufp{i}$ w.r.t. $F_k$ for every $i$,
where $f_k - |\sufp{i}| + 1$ and $f_k - |\suf{i}| + 1] - 1$ are the starting positions of $\sufp{i}$ and $\suf{i}$ in $F_k$, respectively.

\begin{lemma}\label{lem:ao-critical-suffix}
    Assume that $k \geq 7$ is odd and $a \prec b$.
    The previous suffix of $\sufp{i}$ w.r.t. $F_k$ is $\suf{i}$
    for every even $i$ satisfying $4 \leq i \leq k-3$.
\end{lemma}
\begin{proof}
    Since $\sufp{i} = F_i \cdot F_{i-2} \cdots F_4 \cdot ab \cdot \alpha$ for some string $\alpha$,
    $\suf{i}$ is a prefix of $\sufp{i}$.
    Thus, $\suf{i} \prec \sufp{i}$.
    We prove that there is no suffix $x$ of $F_k$ with $\suf{i} \prec x \prec \sufp{i}$ by induction on $i$.

    Let $i = 4$ for the base case.
    Assume on the contrary that there exists a suffix $x$ of $F_k$ such that $\suf{4} \prec x \prec \sufp{4}$.
    Since $\suf{4} = F_5 = F_4 \cdot ab = aba \cdot ab$ and $\sufp{4} = F_4 \cdot F_5 = F_4 \cdot abaab$,
    $x$ can be written as $x = F_4 \cdot abaaa \cdot y$ for some string $y$.
    However, $aaa$ is not a substring of $F_k$ according to Property~\ref{itFibAAA} of Lemma~\ref{lem:fib-properties}, so $x$ cannot exist.
    Thus, the statement holds for the base case.

    Assume that the statement holds for all even $i \leq j$ for some even $j \geq 4$.
    Suppose on the contrary that there exists a suffix $x$ of $F_k$ such that $\suf{j+2} \prec x \prec \sufp{j+2}$.
    Since $\suf{j+2}$ is a prefix of $\sufp{j+2}$,
    we can represent $x$ as $x = \suf{j+2} \cdot y$ for some string $y$.
    There are two cases w.r.t.\ the length of $x$.
    (1) Assume that $|\suf{j+2}| < |x| < |\sufp{j+2}|$.
    Because of the assumption and the fact that $F_{j+2}^2$ is a prefix of $\sufp{j+2}$ by Lemma~\ref{lem:suf-suf}(a),
    $F_{j+2}$ occurs as a substring that is neither prefix nor a suffix of $F_{j+2}^2$.
    However, $F_{j+2}^2$ cannot have such occurrence of $F_{j+2}$ since $F_{j+2}$ is primitive (by Property~\ref{itFibPrimitivity} of Lemma~\ref{lem:fib-properties}),
    a contradiction (by Lemma~\ref{lem:primitive-square}).
    (2) Assume that $|\sufp{j+2}| < |x|$.
    We now use that $\sufp{i} = \suf{i} \cdot \sufp{i-2}$ holds by Lemma~\ref{lem:suf-suf}(b).
    By the assumption, $y$ and $\sufp{j}$ mismatch with $a$ and $b$, respectively
    (since $x$ and $\sufp{j+2}$ have $\suf{j+2}$ as a prefix).
    By Lemma~\ref{lem:suf-suf}, $\sufp{j+2}$ can be represented as
    \begin{equation}\label{eqSufpFactorization}
        \sufp{j+2} = \suf{j+2} \cdot \suf{j} \cdots \suf{j+2-2\ell} \cdot \sufp{j-2\ell}
    \end{equation}
    for some integer $\ell \geq 0$.
    Let $\ell'$ be the largest integer $\ell$ such that the mismatch position is in the factor $\sufp{i-2\ell}$ of the $\sufp{j+2}$-factorization in Eq.~\ref{eqSufpFactorization},
    and $y'$ be the suffix of $y$ that has the factor $\suf{j+2-2\ell'}$ of the factorization in Eq.~\ref{eqSufpFactorization}
    as a prefix (i.e., $y' = y[|\suf{j} \cdots \suf{j+2-2(\ell'-1)}|+1..])$.
    Since $\suf{j+2-2\ell'}$ is a prefix $y'$, and $y'$ and $\sufp{j-2\ell'}$ mismatch at the same position,
    then $\suf{j+2-2\ell'} \prec y' \prec \sufp{j+2-2\ell'}$ holds.
    This fact contradicts the induction hypothesis (see also Fig.~\ref{fig:ao-sketch}).
\end{proof}

\begin{figure}[t]
    \centerline{\includegraphics[width=1.0\linewidth]{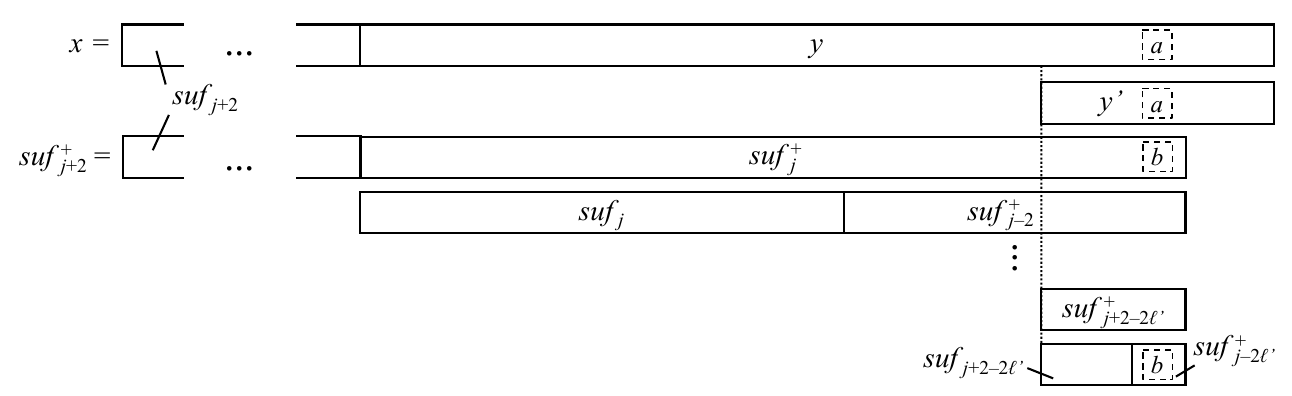}}
    \caption{
        Illustration of the proof of Lemma~\ref{lem:ao-critical-suffix}.
    }\label{fig:ao-sketch}
\end{figure}

Now we can show the following main lemma from the above lemmas.

\begin{lemma}\label{lemLogFactorsOddN}
    Assume that $k \geq 7$ is odd and $a \prec b$.
    Then the lex-parse of $F_k$ is
    \[
        F_k[1..f_{k-1}-2], baF_{k-4}, F_{k-4}, F_{k-6}, \ldots, F_5, a, a, b.
    \]
\end{lemma}
\begin{proof}
    Let $x$ be the longest common prefix of $F_{k-3}$ and $G_{k-3}$.
    Due to Property~\ref{itFibThreeOccs} of Lemma~\ref{lem:fib-properties},
    there are three suffixes
    $F_k = F_{k-2} \cdot x \cdot ba \cdot F_{k-4} \cdot F_{k-5} \cdot F_{k-4}$,
    $\mathit{suf} = F_{k-2} \cdot x \cdot ab$, and
    $F_{k-2}$ that have $F_{k-2}$ as a prefix.
    This implies that $F_{k-2} \prec \mathit{suf} \prec F_k$.
    Thus, the first phrase is $F_{k-2} \cdot x$.
    Then the remaining suffix is $ba \cdot \suf{k-3} = ba \cdot F_{k-2}$.
    We show the second phrase is $ba \cdot F_{k-4}$
    by proving that the previous suffix of $ba \cdot F_{k-2}$ w.r.t. $F_k$ is $ba \cdot F_{k-4}$.
    It is clear that $ba \cdot F_{k-4} \prec ba \cdot F_{k-2}$
    since $ba \cdot F_{k-4}$ is a prefix of $ba \cdot F_{k-2}$.
    Suppose on the contrary that there exists a suffix $y$ of $F_k$
    that satisfies $ba \cdot F_{k-4} \prec y \prec ba \cdot F_{k-2}$.
    From the assumption, $ba \cdot F_{k-4}$ is a prefix of $y$.
    We can observe that there are three occurrences of $ba \cdot F_{k-4}$ in $F_k$ by Lemma~\ref{lem:FibFourOcc} (i.e., the third, the fourth, and the sixth occurrence of $F_{k-4}$ in $F_k$).
    This implies that suffix $ba \cdot F_{k-4} \cdot G_{k-1}$ (regarding the third occurrence of $F_{k-4}$) of $F_k$ is the only candidate of $y$.
    However,
    \[
        y = ba \cdot F_{k-4} \cdot G_{k-1} = ba \cdot F_{k-4} \cdot F_{k-3} \cdot F_{k-2} \succ
        ba \cdot F_{k-4} \cdot G_{k-3} = ba \cdot F_{k-2}
    \]
    holds.
    Thus, the previous suffix of $ba \cdot F_{k-2}$ w.r.t. $F_k$ is $ba \cdot F_{k-4}$,
    and the second phrase is $ba \cdot F_{k-4}$.
    Then, the remaining suffix is $\sufp{k-5}$.
    From Lemma~\ref{lem:ao-critical-suffix},
    the next phrase is $\suf{k-5} = F_{k-4}$ and the remaining suffix is $\sufp{k-7}$.
    This continues until the remaining suffix is $aab$.
    It is easy to see that the last three phrases are $a, a, b$.
\end{proof}

We can also prove the following lemmas similarly.

\begin{lemma}\label{lemFourFactorsOddN}
    Assume that $k \geq 7$ is odd and $b \prec a$.
    Then the lex-parse of $F_k$ is
    \[
        F_{k-2}, F_k[f_{k-2}+1..f_k-2], a, b.
    \]
\end{lemma}

\begin{lemma}\label{lemLogFactorsEvenN}
    Assume that $k \geq 6$ is even and $b \prec a$.
    Then the lex-parse of $F_k$ is
    \[
        F_k[1..f_{k-1}-2], abF_{k-4}, F_{k-4}, F_{k-6}, \ldots, F_6, b, a.
    \]
\end{lemma}

Overall, Theorem~\ref{thmAOLowerBound} holds.
\section*{Acknowledgments}
This work was supported by JSPS KAKENHI Grant Numbers 
JP21K17705, JP23H04386 (YN), 
JP21K17701, JP23H04378 (DK),
JP22H03551 (SI), 
and JP20H04141 (HB).
\bibliographystyle{abbrv}
\bibliography{ref}

\begin{thebibliography}{10}

\bibitem{AKAGI2023104999}
T.~Akagi, M.~Funakoshi, and S.~Inenaga.
\newblock Sensitivity of string compressors and repetitiveness measures.
\newblock {\em Information and Computation}, 291:104999, 2023.

\bibitem{BentleyGT20}
J.~W. Bentley, D.~Gibney, and S.~V. Thankachan.
\newblock On the complexity of {BWT}-runs minimization via alphabet reordering.
\newblock In F.~Grandoni, G.~Herman, and P.~Sanders, editors, {\em 28th Annual
  European Symposium on Algorithms, {ESA} 2020, September 7-9, 2020, Pisa,
  Italy (Virtual Conference)}, volume 173 of {\em LIPIcs}, pages 15:1--15:13.
  Schloss Dagstuhl - Leibniz-Zentrum f{\"{u}}r Informatik, 2020.

\bibitem{blumer87complete}
A.~Blumer, J.~Blumer, D.~Haussler, R.~M. McConnell, and A.~Ehrenfeucht.
\newblock Complete inverted files for efficient text retrieval and analysis.
\newblock {\em J. {ACM}}, 34(3):578--595, 1987.

\bibitem{Burrows94ablock-sorting}
M.~Burrows and D.~Wheeler.
\newblock A block-sorting lossless data compression algorithm.
\newblock Technical report, DIGITAL SRC RESEARCH REPORT, 1994.

\bibitem{DBLP:conf/dcc/CenzatoGLR23}
D.~Cenzato, V.~Guerrini, Z.~Lipt{\'{a}}k, and G.~Rosone.
\newblock Computing the optimal {BWT} of very large string collections.
\newblock In A.~Bilgin, M.~W. Marcellin, J.~Serra{-}Sagrist{\`{a}}, and J.~A.
  Storer, editors, {\em Data Compression Conference, {DCC} 2023, Snowbird, UT,
  USA, March 21-24, 2023}, pages 71--80. {IEEE}, 2023.

\bibitem{ChenFL58:_lyndon_factorization_}
K.~T. Chen, R.~H. Fox, and R.~C. Lyndon.
\newblock Free differential calculus. {IV. The} quotient groups of the lower
  central series.
\newblock {\em Annals of Mathematics}, 68(1):81--95, 1958.

\bibitem{DBLP:journals/ipl/ClareD19}
A.~Clare and J.~W. Daykin.
\newblock Enhanced string factoring from alphabet orderings.
\newblock {\em Inf. Process. Lett.}, 143:4--7, 2019.

\bibitem{DBLP:conf/gecco/ClareDMZ19}
A.~Clare, J.~W. Daykin, T.~Mills, and C.~Zarges.
\newblock Evolutionary search techniques for the {Lyndon} factorization of
  biosequences.
\newblock In M.~L{\'{o}}pez{-}Ib{\'{a}}{\~{n}}ez, A.~Auger, and
  T.~St{\"{u}}tzle, editors, {\em Proceedings of the Genetic and Evolutionary
  Computation Conference Companion, {GECCO} 2019, Prague, Czech Republic, July
  13-17, 2019}, pages 1543--1550. {ACM}, 2019.

\bibitem{DBLP:journals/algorithmica/CrochemoreR95}
M.~Crochemore and W.~Rytter.
\newblock Squares, cubes, and time-space efficient string searching.
\newblock {\em Algorithmica}, 13(5):405--425, 1995.

\bibitem{crochemore97direct}
M.~Crochemore and R.~V{\'{e}}rin.
\newblock Direct construction of compact directed acyclic word graphs.
\newblock In {\em Proc.\ CPM}, volume 1264 of {\em LNCS}, pages 116--129, 1997.

\bibitem{DBLP:conf/cwords/FujimaruNI23}
H.~Fujimaru, Y.~Nakashima, and S.~Inenaga.
\newblock On sensitivity of compact directed acyclic word graphs.
\newblock In A.~E. Frid and R.~Mercas, editors, {\em Combinatorics on Words -
  14th International Conference, {WORDS} 2023, Ume{\aa}, Sweden, June 12-16,
  2023, Proceedings}, volume 13899 of {\em Lecture Notes in Computer Science},
  pages 168--180. Springer, 2023.

\bibitem{giancarlo23new}
R.~Giancarlo, G.~Manzini, A.~Restivo, G.~Rosone, and M.~Sciortino.
\newblock A new class of string transformations for compressed text indexing.
\newblock {\em Inf. Comput.}, 294:105068, 2023.

\bibitem{GibneyT21}
D.~Gibney and S.~V. Thankachan.
\newblock Finding an optimal alphabet ordering for {Lyndon} factorization is
  hard.
\newblock In M.~Bl{\"{a}}ser and B.~Monmege, editors, {\em 38th International
  Symposium on Theoretical Aspects of Computer Science, {STACS} 2021, March
  16-19, 2021, Saarbr{\"{u}}cken, Germany (Virtual Conference)}, volume 187 of
  {\em LIPIcs}, pages 35:1--35:15. Schloss Dagstuhl - Leibniz-Zentrum f{\"{u}}r
  Informatik, 2021.

\bibitem{DBLP:conf/dlt/GiulianiILRSU23}
S.~Giuliani, S.~Inenaga, Z.~Lipt{\'{a}}k, G.~Romana, M.~Sciortino, and
  C.~Urbina.
\newblock Bit catastrophes for the {Burrows--Wheeler} transform.
\newblock In F.~Drewes and M.~Volkov, editors, {\em Developments in Language
  Theory - 27th International Conference, {DLT} 2023, Ume{\aa}, Sweden, June
  12-16, 2023, Proceedings}, volume 13911 of {\em Lecture Notes in Computer
  Science}, pages 86--99. Springer, 2023.

\bibitem{DBLP:journals/tcs/INIBT16}
T.~I, Y.~Nakashima, S.~Inenaga, H.~Bannai, and M.~Takeda.
\newblock Faster {Lyndon} factorization algorithms for {SLP} and {LZ78}
  compressed text.
\newblock {\em Theor. Comput. Sci.}, 656:215--224, 2016.

\bibitem{DBLP:journals/mst/InoueMNIBT22}
H.~Inoue, Y.~Matsuoka, Y.~Nakashima, S.~Inenaga, H.~Bannai, and M.~Takeda.
\newblock Factorizing strings into repetitions.
\newblock {\em Theory Comput. Syst.}, 66(2):484--501, 2022.

\bibitem{KempaP18}
D.~Kempa and N.~Prezza.
\newblock At the roots of dictionary compression: string attractors.
\newblock In {\em {STOC} 2018}, pages 827--840, 2018.

\bibitem{KopplI15}
D.~K{\"{o}}ppl and T.~I.
\newblock Arithmetics on suffix arrays of {Fibonacci} words.
\newblock In F.~Manea and D.~Nowotka, editors, {\em Combinatorics on Words -
  10th International Conference, {WORDS} 2015, Kiel, Germany, September 14-17,
  2015, Proceedings}, volume 9304 of {\em Lecture Notes in Computer Science},
  pages 135--146. Springer, 2015.

\bibitem{DBLP:conf/soda/LagardeP18}
G.~Lagarde and S.~Perifel.
\newblock Lempel-{Ziv}: a "one-bit catastrophe" but not a tragedy.
\newblock In A.~Czumaj, editor, {\em Proceedings of the Twenty-Ninth Annual
  {ACM-SIAM} Symposium on Discrete Algorithms, {SODA} 2018, New Orleans, LA,
  USA, January 7-10, 2018}, pages 1478--1495. {SIAM}, 2018.

\bibitem{Lothaire83}
M.~Lothaire.
\newblock {\em Combinatorics on Words}.
\newblock Addison-Wesley, 1983.

\bibitem{lothaire2005applied}
M.~Lothaire.
\newblock {\em Applied combinatorics on words}, volume 105.
\newblock Cambridge University Press, 2005.

\bibitem{lyndon54:_burnside}
R.~C. Lyndon.
\newblock On {B}urnside's problem.
\newblock {\em Transactions of the American Mathematical Society}, 77:202--215,
  1954.

\bibitem{DBLP:conf/ppsn/MajorCDMGZ20}
L.~Major, A.~Clare, J.~W. Daykin, B.~Mora, L.~J.~P. Gamboa, and C.~Zarges.
\newblock Evaluation of a permutation-based evolutionary framework for {Lyndon}
  factorizations.
\newblock In T.~B{\"{a}}ck, M.~Preuss, A.~H. Deutz, H.~Wang, C.~Doerr, M.~T.~M.
  Emmerich, and H.~Trautmann, editors, {\em Parallel Problem Solving from
  Nature - {PPSN} {XVI} - 16th International Conference, {PPSN} 2020, Leiden,
  The Netherlands, September 5-9, 2020, Proceedings, Part {I}}, volume 12269 of
  {\em Lecture Notes in Computer Science}, pages 390--403. Springer, 2020.

\bibitem{DBLP:journals/corr/abs-2401-16435}
L.~Major, A.~Clare, J.~W. Daykin, B.~Mora, and C.~Zarges.
\newblock Heuristics for the run-length encoded {Burrows--Wheeler} transform
  alphabet ordering problem.
\newblock {\em CoRR}, abs/2401.16435, 2024.

\bibitem{MELANCON2000137}
G.~Melançon.
\newblock Lyndon factorization of sturmian words.
\newblock {\em Discrete Mathematics}, 210(1):137--149, 2000.

\bibitem{DBLP:journals/corr/abs-2004-02781}
G.~Navarro.
\newblock Indexing highly repetitive string collections.
\newblock {\em CoRR}, abs/2004.02781, 2020.

\bibitem{repetitiveness_Navarro21a}
G.~Navarro.
\newblock Indexing highly repetitive string collections, part {I:}
  repetitiveness measures.
\newblock {\em {ACM} Comput. Surv.}, 54(2):29:1--29:31, 2021.

\bibitem{NavarroOP21}
G.~Navarro, C.~Ochoa, and N.~Prezza.
\newblock On the approximation ratio of ordered parsings.
\newblock {\em {IEEE} Trans. Inf. Theory}, 67(2):1008--1026, 2021.

\bibitem{StorerS78}
J.~A. Storer and T.~G. Szymanski.
\newblock The macro model for data compression (extended abstract).
\newblock In R.~J. Lipton, W.~A. Burkhard, W.~J. Savitch, E.~P. Friedman, and
  A.~V. Aho, editors, {\em Proceedings of the 10th Annual {ACM} Symposium on
  Theory of Computing, May 1-3, 1978, San Diego, California, {USA}}, pages
  30--39. {ACM}, 1978.

\bibitem{LZ77}
J.~Ziv and A.~Lempel.
\newblock A universal algorithm for sequential data compression.
\newblock {\em IEEE Transactions on Information Theory}, 23(3):337--343, 1977.

\bibitem{LZ78}
J.~Ziv and A.~Lempel.
\newblock Compression of individual sequences via variable-rate coding.
\newblock {\em {IEEE} Trans. Inf. Theory}, 24(5):530--536, 1978.

\end{thebibliography}
\end{document}